\documentclass[letterpaper,11pt]{article}

\usepackage{amsfonts}
\usepackage{indentfirst}
\usepackage{amssymb}
\usepackage{amsmath}
\usepackage{eufrak}
\usepackage{amsthm}
\usepackage{hyperref}
\usepackage{cite}
\usepackage{amscd}

\newtheorem{theorem}{Theorem}[section]
\newtheorem{lemma}[theorem]{Lemma}
\newtheorem{proposition}[theorem]{Proposition}
\newtheorem{corollary}[theorem]{Corollary}
\newtheorem{remark}{Remark}[section]
\newtheorem{definition}{Definition}[section]

\title{Simultaneous Deformations of Lie Algebroids and Lie Subalgebroids}
\author{Xiang Ji\\Department of Mathematics\\The Pennsylvania State University}
\date{September 30, 2012}

\begin{document}

\maketitle

\begin{abstract}
$L_\infty$-algebra is an algebraic structure suitable for describing deformation problems. In this paper we construct two $L_\infty$-algebras, one to control the deformations of Lie algebroids, and the other to control the deformations of Lie subalgebroids. We also combine these two $L_\infty$-algebras into one to control the simultaneous deformations of a Lie algebroid and its Lie subalgebroids. The results generalize the deformation theory of Lie algebra and Lie subalgebras. Applications of our results include deformations of foliations, deformations of complex structures and deformations of homomorphisms of Lie algebroids.
\end{abstract}

\tableofcontents
\hypersetup{linktocpage}

\section{Introduction}

\par Deformation theories arose from the middle of the 20th century, first originated from deforming algebraic structures like Lie algebras, and then developed into the field of geometry to deform geometric structures such as Riemann structures, Poisson(symplectic) structures and so on (see \cite{pommaret} for an introduction). $L_\infty$-algebras turn out to be the right language for describing deformation theories, which include differential graded Lie algebras as special cases. An $L_\infty$-algebra controls or governs the deformation if there is a one-one correspondence between the deformations and the Maurer-Cartan elements of the $L_\infty$-algebra. The reader may find introduction to $L_\infty$-algebras in \cite{stasheff} and \cite{schatz1}. 

\par First let us recall the definition of Lie algebroids and Lie subalgebroids. A \textsl{Lie algebroid} is a vector bundle $A$ over some smooth manifold $M$ together with a Lie bracket $[\cdot,\cdot]$ on the set $\Gamma(A)$ of sections of $A$ and a bundle map $\rho:A\to TM$, called the anchor, such that
\begin{itemize}
\item $\rho:\Gamma(A)\to \mathfrak{X}(M)$ is a homomorphism of Lie algebras,
\item $[v,f\cdot w]=f[v,w]+(\rho(v)f)\cdot w$ for any $v, w \in\Gamma(A)$ and $f\in\mathcal{C}^{\infty}(M)$.
\end{itemize}
When the base manifold $M$ is a point, the Lie algebroid $A$ becomes a Lie algebra. Also the tangent bundle $TM$ is a Lie algebroid with the Lie bracket of vector fields and the anchor $\rho=Id_{TM}:TM\to TM$. A \textsl{Lie subalgebroid} of $A$ is a subbundle $E\subset A$ over some closed submanifold $S\subset M$, such that 
\begin{itemize}
\item $\rho(E)\subset TS$,
\item for any $e_1,e_2\in \Gamma(E)$ and any extension $\tilde{e}_1,\tilde{e}_2\in\Gamma(A)$, $[\tilde{e}_1,\tilde{e}_2]|_S\in \Gamma(E)$ and it does not depend on the choice of extensions.
\end{itemize} 
When $M$ is a point, the concept of Lie subalgebroids coincides with the concept of Lie subalgebras. For any submanifold $S\subset M$, $TS$ is always a Lie subalgebroid of $TM$. Other examples can be found in \cite{weinstein} and \cite{mackenzie}.

\par We explain the definition of $L_\infty$-algebras in Section \ref{sec2}. A method invented by T. Voronov to construct $L_\infty$-algebras (\cite{voronov1}, \cite{voronov2}) which will be used later in the paper is also introduced. In order to describe the simultaneous deformation, we also include a result of Y. Fr$\acute{e}$gier and M. Zambon (\cite{zambon}). 

\par To solve the deformation problem of Lie algebroids and Lie subalgebroids, we are going to find two $L_\infty$-algebras, one to govern the deformations of Lie algebroids and the other to govern the deformations of Lie subalgebroids. These are the major parts of our paper (Theorem \ref{thm1} and Theorem \ref{thm2}). Our results generalize the deformation theory of Lie algebras and Lie subalgebras in \cite{zambon1}. The $L_\infty$-algebra governing the deformations of a Lie algebroid turns out to be a differential graded Lie algebra. M. Crainic and I. Moerdijk get similar results in \cite{crainic}, but our method saves a lot of arguments. We also describe some alternative approaches to solve the two deformation problem using different viewpoints in these sections.

\par In Section \ref{sec5}, we combine the two $L_\infty$-algebras into one to govern the simultaneous deformations of a Lie algebroid and its Lie subalgebroid. This is Corollary \ref{cor2}. In the last section (Section \ref{sec6}), we give some applications to our results including deformations of foliations, deformations of complex structures and deformations of homomorphisms of Lie algebroids.

\par \textbf{Acknowledgments.} I would endlessly thank my advisor Prof. Ping Xu for his encouragements and a lot of helpful as well as insightful discussions and suggestions, especially in interchanging different viewpoints to attack the problem. I also want to thank Rajan Mehta for discussions and useful information supplied on super languages. 

\section{$L_\infty$-algebras and Maurer-Cartan elements}
\label{sec2}
\subsection{Definition of $L_\infty$-algebras}
\par Whenever describing deformations of objects, no matter algebraic or geometric, the $L_\infty$-algebra is always the right language to adopt. The $L_\infty$-algebras can be viewed as generalizations of differential graded Lie algebras.

\begin{definition}
A \textsl{differential graded Lie algebra}(DGLA) is a $\mathbb{Z}$-graded vector space $V=\oplus_{n\in \mathbb{Z}} V_n$ over $\mathbb{R}$ together with a graded antisymmetric bilinear bracket $[\cdot,\cdot]:V_i\otimes V_j\to V_{i+j}$ and a differential $d: V\to V$ satisfying
\begin{itemize}
\item graded antisymmetry: $[a,b]=-(-1)^{|a|\cdot |b|}[b,a]$,
\item Jacobi identity: $[a,[b,c]]=[[a,b],c]+(-1)^{|a|\cdot |b|}[b,[a,c]]$,
\item Leibniz rule: $d[a,b]=[da,b]+(-1)^{|a|}[a,db]$,
\item differential condition: $d^2=0$ and $|da|=|a|+1$,
\end{itemize}
for homogeneous $a,b,c\in V$. Here $|a|$ denotes the degree of $a$, i.e. $|a|=n$ if $a\in V_n$.
\end{definition}

\par The concept of $L_\infty$-algebra is first introduced by Lada and Stasheff \cite{stasheff}. There are several equivalent definitions. In the following, we shall list some of them, and indicate their equivalences. The first one follows Lada and Stasheff and the rule of \textsl{Koszul sign} is needed. Given a graded vector space $V$, there are two natural actions of the symmetric group $S_n$ on $\otimes ^n V$. The symmetric action is defined by
\[
	\sigma(a_1\otimes \cdots \otimes a_n)=(-1)^{|a_j||a_{j+1}|}a_1\otimes \cdots a_{j+1}\otimes a_j \otimes \cdots \otimes a_n
\]
where $\sigma\in S_n$ is the transposition of the $j$-th and the $(j+1)$-th element. This extends to a well-defined action of $S_n$ on $\otimes^n V$. In general we use $e(\tau,a_1,\cdots,a_n)$ to denote the sign resulting from the action of a general $\tau\in S_n$ on $a_1\otimes\cdots\otimes a_n$ with $a_1,\cdots,a_n$ homogeneous, i.e.
\[
	\tau(a_1\otimes \cdots \otimes a_n)=e(\tau,a_1,\cdots,a_n)a_{\tau(1)}\otimes\cdots\otimes a_{\tau(n)}.
\]
When the elements $a_1, \cdots, a_n$ are clear from the context, we abbreviate $e(\tau,a_1,\cdots,a_n)$ to be $e(\tau)$. The skew-symmetric action is defined by
\[
	\tau(a_1\otimes \cdots \otimes a_n)=(-1)^{\tau}e(\tau)a_{\tau(1)}\otimes\cdots\otimes a_{\tau(n)},
\]
where $(-1)^{\tau}$ is $1$ if $\tau$ is even and $-1$ if $\tau$ is odd. 
\par A multi-linear map $l_n:\otimes^n V\to V$ is symmetric if 
\[
	l_n(a_{\tau(1)}\otimes\cdots\otimes a_{\tau(n)})=e(\tau)l_n(a_1\otimes\cdots\otimes a_n)
\]
for arbitrary homogeneous elements $a_1,\cdots,a_n\in V$. A multi-linear map $m_n:\otimes^n V\to V$ is antisymmetric if 
\[
	l_n(b_{\tau(1)}\otimes\cdots\otimes b_{\tau(n)})=(-1)^{\tau}e(\tau)l_n(b_1\otimes\cdots\otimes b_n)
\]
for arbitrary homogeneous elements $b_1,\cdots,b_n\in V$.

\par In the following we use $S_{i,n-i}$ to denote the collection of all $(i,n-i)$ shuffles in $S_n$, i.e. $\tau\in S_{i,n-i}$ if and only if $\tau(1)<\cdots<\tau(i)$ and $\tau(i+1)<\cdots<\tau(n)$.

\begin{definition}
An $L_\infty$-algebra is a $\mathbb{Z}$-graded vector space $V$ over $\mathbb{R}$ together with a family of antisymmetric multi-linear maps $l_k:\otimes^n V\to V$ of degree $2-k$ ($k\ge 0$) such that for any homogeneous $a_1,\cdots,a_n\in V$,
\[
	\sum_{i+j=n+1} (-1)^{i(j-1)}\sum_{\tau \in S_{i,n-i}}(-1)^{\tau}e(\tau)l_i(l_j(a_{\tau(1)},\cdots,a_{\tau(j)}),a_{\tau(j+1)},\cdots,a_{\tau(n)})=0.
\]
\end{definition}

\par An $L_\infty$-algebra $(V,\{l_k\})$ is called flat if $l_0=0$. A flat $L_\infty$-algebra becomes a DGLA if $l_k=0$ for $k\ge 3$.

\par In the above definition, if shifting the degree of $V$ by $1$, all antisymmetric multi-linear maps $l_k$'s become symmetric multi-linear maps of degree $1$. In this case, the definition of $L_\infty$-algebras should be modified as follows.

\begin{definition}
\label{l1}
An $L_\infty[1]$-algebra is a $\mathbb{Z}$-graded vector space $V$ over $\mathbb{R}$ together with a family of symmetric multi-linear maps $m_k:\otimes^n V\to V$ of degree $1$ ($k\ge 0$) such that for any homogeneous $b_1,\cdots,b_n\in V$,
\[
	\sum_{i+j=n+1} \sum_{\tau \in S_{i,n-i}}e(\tau)m_i(m_j(b_{\tau(1)},\cdots,b_{\tau(j)}),b_{\tau(j+1)},\cdots,b_{\tau(n)})=0.
\]
\end{definition}

\par Given a $\mathbb{Z}$-graded vector space $V$ and $n\in \mathbb{Z}$, $V[n]$ is the graded vector space defined by $(V[n])_k=V_{n+k}$. 

\begin{remark}
An $L_\infty$-algebra structure on $V$ is equivalent to an $L_\infty[1]$-algebra structure on $V[1]$ by the shift isomorphism:
\begin{align*}
	(\otimes^n V)[n] &\to \otimes^n (V[1])\\
a_1\otimes\cdots\otimes a_n & \mapsto (-1)^{(n-1)|a_1|+(n-2)|a_2|+\cdots+|a_{n-1}|}a_1\otimes\cdots\otimes a_n.
\end{align*}
\end{remark}

\par Using super language, $L_\infty$-algebra is described more elegantly in \cite{shoikhet}. The details of super and graded manifolds can be found in \cite{schatz}.

\begin{definition}[\cite{shoikhet}]
An $L_\infty[1]$-algebra is a $\mathbb{Z}$-graded vector space $V$ and an odd vector field $X_Q$ of degree $1$ on the super vector space $V[1]$ such that $[X_Q,X_Q]=0$.
\end{definition}

\begin{remark}
Given an $L_\infty[1]$-algebra $(V,\{m_k\})$, the super vector field $X_Q$ is determined by
\[
	X_Q: \mathcal{C}^{\infty}(V[1])=\oplus_{n=0}^{\infty}\wedge^n (V^*) \to \mathcal{C}^{\infty}(V[1]),
\] 
such that for any $\eta\in V^*\subset\mathcal{C}^{\infty}(V[1])$, $X_Q(\eta)\in \mathcal{C}^{\infty}(V[1])$ and any $k\ge 0$, $v_1,\cdots,v_k \in V$
\[
	X_Q(\eta)(v_1,\cdots,v_k)=<m_k(v_1,\cdots,v_n),\eta>
\]
where $<\cdot,\cdot>$ denotes the canonical pairing between $V$ and $V^*$. Conversely, given a vector field $X_Q$ of degree $1$ on $V[1]$, the multi-linear maps $m_k$'s can be recovered from the above formula as well. The condition $[X_Q,X_Q]=0$ is equivalent to the vanishing conditions in Definition.\ref{l1}.
\end{remark}

\par Stasheff (\cite{stasheff}) also gives a conceptual approach to define $L_\infty$-algebras. Given a $\mathbb{Z}$-graded vector space $V$, the tensor algebra $\mathcal{T}(V)=\sum_{k=0}^{\infty}\otimes^{k} V$ is a coassociative coalgebra with the comultiplication $\Delta: \mathcal{T}(V)\to \mathcal{T}(V) \otimes \mathcal{T}(V)$ defined by
\[
	\Delta(x_1\otimes\cdots\otimes x_n)=\sum_{k=0}^{n}(x_1\otimes\cdots\otimes x_k)\otimes(x_{k+1}\otimes\cdots\otimes x_n)
\]
Let $\mathcal{S}(V)$ be the symmetric product of $V$, i.e. the subspace of $\mathcal{T}(V)$ consisting of elements of the form
\[
	v_1\odot\cdots\odot v_n=\sum_{\tau\in S_n}e(\tau,v_1,\cdots,v_n)v_{\tau(1)}\otimes\cdots\otimes v_{\tau(n)}.
\]
$\mathcal{S}(V)$ is a subcoalgebra of $\mathcal{T}(V)$. Any symmetric multi-linear map can automatically define over $\mathcal{S}(V)$. Suppose the maps $\{m_k:\otimes^k V\to V\}$ make $V$ into an $L_\infty[1]$-algebra. They yield a coderivation $M$ of degree 1 on the coalgebra $(\mathcal{S}(V),\Delta)$ by
\[
	M(x_1\odot\cdots\odot x_n)=\sum_{r+s=n}\sum_{\tau\in S_{r,s}}e(\tau)m_r(x_{\tau(1)}\odot\cdots\odot x_{\tau(r)})\odot x_{\tau(r+1)}\odot\cdots\odot x_{\tau(n)}
\]
which satisfies $\Delta\circ M=(M\otimes Id+ Id\otimes M)\circ \Delta$. The vanishing condition in Definition.\ref{l1} is equivalent to $M\circ M=0$. A coderivation of degree 1 satisfying $M\circ M=0$ is called a codifferential.

\begin{definition}
A $\mathbb{Z}$-graded vector space $V$ is an $L_\infty[1]$-algebra if there is a codifferential $M$ on the coassociative subcoalgebra $(\mathcal{S}(V),\Delta)$ .
\end{definition}

\par All the above four definitions of $L_\infty$-algebras are equivalent. In the following we mainly apply Definition \ref{l1}. For any $v\in V$, we use $v[1]$ to denote the corresponding element in $V[1]$.

\begin{remark} 
\label{cohomology}
Given a flat $L_\infty[1]$-algebra $(V,\{m_k\})$, $m_1$ is a differential of $V$ according to its vanishing axioms. $(V,m_1)$ becomes a cochain complex and the cohomology $H^*(V,m_1)$ turns out to be a Lie algebra with the bracket induced by $m_2:\otimes^2 V\to V$.
\end{remark}

\subsection{$L_\infty$-algebras by higher derived brackets}

\par T. Voronov provides an approach to construct $L_\infty$-algebras via higher derived brackets.

\begin{definition}
A \textsl{V-algebra} is a graded Lie algebra $(V,[\cdot,\cdot])$, which is the direct sum of two Lie subalgebras $V=\mathfrak{a}\oplus \mathfrak{b}$ and one of them $\mathfrak{a}$ is abelian.
\end{definition}

\par We denote the projection $V\to \mathfrak{a}$ by $P$, then $\mathfrak{b}=Ker(P)$.

\begin{theorem}[\cite{voronov1}]
\label{v1}
Suppose $(V,[\cdot,\cdot],\mathfrak{a},P)$ is a V-algebra and $\Delta\in V$ a homogeneous element of degree $1$ which preserves $Ker(P)$ (i.e. $[\Delta,\mathfrak{b}]\subset \mathfrak{b}$) and satisfies $[\Delta,\Delta]=0$, then $\mathfrak{a}$ becomes an $L_\infty[1]$-algebra with structure maps $m_0:\mathbb{R}\to V_1$ mapping $1$ to $P(\Delta)$ and
\begin{eqnarray}
	m_k(a_1,\cdots,a_k)=P[\cdots[[\Delta,a_1],a_2],\cdots,a_k]
\end{eqnarray}
for $k\ge 1$.
\end{theorem}

\begin{remark}
The structure maps $m_k$'s are called the higher derived brackets of $\Delta$. The element $\Delta$ in the theorem is called a Maurer-Cartan element of the V-algebra. Particularly, if $\Delta$ lies in $Ker(P)$, then $[\Delta,Ker(P)]\subset Ker(P)$ is automatic, and $\Delta$ is a Maurer-Cartan element if and only if $[\Delta,\Delta]=0$. In this case, the corresponding $L_\infty[1]$-algebra $\mathfrak{a}$ is flat. The original theorem is valid for any derivation $D:V\to V$ of degree $1$ which preserves $Ker(P)$ and satisfies $D\circ D=0$. Here we only use the special case of an inner derivation $D=[\Delta,\cdot]$.
\end{remark}

\par Given the conditions in the above remark, there is also an $L_\infty[1]$-algebra structure on $V[1]\oplus \mathfrak{a}$.

\begin{theorem}[\cite{voronov2}]
\label{v2}
Suppose $(V,[\cdot,\cdot],\mathfrak{a},P)$ is a V-algebra and $\Delta\in Ker(P)_1$ satisfying $[\Delta,\Delta]=0$, then $V[1]\oplus\mathfrak{a}$ becomes a flat $L_\infty[1]$-algebra with structure maps
\begin{eqnarray}
	m_1(v[1],a)=(-[\Delta,v][1],P(v+[\Delta,a])),\\
	m_2(v[1],w[1])=(-1)^{|v|}[v,w][1],\\
	m_k(v[1],a_1,\cdots,a_{k-1})=P[\cdots[[v,a_1],a_2],\cdots,a_{k-1}],\\
	m_k(a_1,\cdots,a_k)=P[\cdots[[\Delta,a_1],a_2],\cdots,a_k],
\end{eqnarray}
for $k\ge 2$, and all other combinations vanish. Here $a_i\in\mathfrak{a}$ and $v,w$ are homogeneous elements in $V$.
\end{theorem}

\par We use $\mathfrak{a}_{\Delta}^P$ and $(V[1]\oplus\mathfrak{a})_{\Delta}^{P}$ to denote the $L_\infty[1]$-algebras constructed in the previous theorems.

\subsection{Maurer-Cartan elements and twisted $L_\infty$-algebras}

\par To describe deformations by $L_\infty$-algebras, the idea is to let deformations correspond to Maurer-Cartan elements of the $L_\infty$-algebras.

\begin{definition}
A formal Maurer-Cartan element in an $L_\infty[1]$-algebra $(V,\{m_k\})$ is an element $v$ of degree $0$ satisfying the Maurer-Cartan equation
\begin{eqnarray}
	\sum_{k=0}^{\infty}\frac{1}{k!}m_k(v,\cdots,v)=0.
\end{eqnarray}
\end{definition}
\par Here we encounter the convergence problem. If the structure maps $m_k=0$ for $k$ large, the problem of convergence can be avoided. Otherwise we need make the infinite sum well-defined. In our paper we will use the analyticity condition to ensure the convergence.

\par It turns out the $L_\infty$-algebra structure can be twisted by its Maurer-Cartan elements (\cite{getzler}). Y. Fr$\acute{e}$gier and M. Zambon also show in \cite{zambon} how this works under T.Voronov's construction. 

\par To describe this, we need the following notation. Given a V-algebra $(V,[\cdot,\cdot],\mathfrak{a},P)$ and $\phi\in \mathfrak{a}_0$, we can define a linear map $P_{\phi}=P\circ e^{[\cdot, \phi]}:V\to \mathfrak{a}$ by 
\begin{eqnarray}
\label{exponents}
	v\mapsto \sum_{n=0}^{\infty}\frac{[\cdots[v,\overbrace{\phi],\cdots,\phi]}^{n \text{ times}}}{n!}.
\end{eqnarray}

Given $\Delta\in Ker(P)_1$ s.t. $[\Delta,\Delta]=0$ (i.e. $\Delta$ is a Maurer-Cartan element of the V-algebra), then $\phi\in \mathfrak{a}_0$ is a Maurer-Cartan element of $\mathfrak{a}_{\Delta}^P$ if and only if $P_{\phi}(\Delta)=0$. Here the convergence is also a problem. To make the map $P_{\phi}$ meaningful, as this moment, we require $\phi\in \mathfrak{a}_0$  being nilpotent, i.e. $(ad_\phi)^n=0$ when $n$ is large.

\begin{theorem}[\cite{zambon}]
\label{zam1}
Suppose $(V,[\cdot,\cdot],\mathfrak{a},P)$ is a V-algebra and $\Delta\in Ker(P)_1$ satisfies $[\Delta,\Delta]=0$. If $\phi$ is a Maurer-Cartan element of $\mathfrak{a}_{\Delta}^P$ which is nilpotent, then 
\begin{itemize}
\item $(V,[\cdot,\cdot],\mathfrak{a},P_{\phi})$ forms a V-algebra and $\mathfrak{a}_{\Delta}^{P_{\phi}}$ is an $L_\infty[1]$-algebra,
\[
	\tilde{\phi} \text{ is a Maurer-Cartan element of } \mathfrak{a}_{\Delta}^{P_{\phi}}\Leftrightarrow \phi+\tilde{\phi} \text{ is a Maurer-Cartan element of } \mathfrak{a}_{\Delta}^P;
\]
\item for any $\tilde{\Delta}\in Ker(P)_1$ and $\tilde{\phi}\in \mathfrak{a}_0$
\[
	(\tilde{\Delta},\tilde{\phi}) \text{ is a MC-element of } (V[1]\oplus\mathfrak{a})_{\Delta}^{P_{\phi}}\Leftrightarrow 
	\left\{
		\begin{array}{l}
			[\Delta+\tilde{\Delta},\Delta+\tilde{\Delta}]=0\\
			\phi+\tilde{\phi} \text{ is a MC element of } \mathfrak{a}_{\Delta+\tilde{\Delta}}^P
		\end{array}
	\right..
\]
\end{itemize}
\end{theorem}
\begin{remark}
\label{convergence}
In this theorem, the nilpotent condition can be removed if the convergence of the Maurer-Cartan equation is guaranteed in advance. Specially in our paper, we will use the analyticity condition to ensure the convergence and in this case the theorem still holds without requiring $\phi$ nilpotent.
\end{remark}
\par We will use this construction to control the simultaneous deformation of Lie algebroids and Lie subalgebroids in Section \ref{sec5}.

\section{Deformations of Lie algebroids}
\label{sec3}
\par In this section we come to the problem of deforming a Lie algebroid. We will construct a differential graded Lie algebra to control the deformation. Given a Lie algebroid $(A,[\cdot,\cdot],\rho)$ over $M$, consider the corresponding supermanifold $A[1]$. The set of smooth functions on $A[1]$,
\[
	\mathcal{C}^\infty(A[1])=\Gamma(\sum_{n=0}^{\infty}\wedge^n(A[1])^*)
\]
is a graded commutative associative algebra. Here $(A[1])^*$ is a graded vector bundle in which the degree $1$ component is $A^*$ and all other components are trivial. Denote the set of all super vector fields over $A[1]$ by $\mathfrak{X}(A[1])$. $X\in \mathfrak{X}(A[1])$ can be treated as a derivation of smooth functions, i.e.
\[
	X: \mathcal{C}^\infty(A[1])\to \mathcal{C}^\infty(A[1]),
\] 
satisfying $X(f\cdot g)=(X(f))\cdot g+(-1)^{|X|\cdot|f|}f\cdot X(g)$. Here $X$ is homogeneous of degree $|X|$ and $f$ is homogeneous of degree $|f|$. Similar to the nonsuper case, $\mathfrak{X}(A[1])$ is a graded Lie algebra with respect to the bracket
\[
	[[X,Y]]=X\circ Y-(-1)^{|X|\cdot |Y|}Y\circ X.
\]

\par Given the Lie algebroid $(A,[\cdot,\cdot],\rho])$, we can define a super vector field $X_Q$ by the following: for any $\xi\in \mathcal{C}^{\infty}(A[1])$ of degree $n$, $X_Q(\xi)\in\mathcal{C}^{\infty}(A[1])$ is of degree $n+1$ and determined by
\begin{eqnarray}
\label{X_Qdef}
	X_Q(\xi)(v_0,\cdots,v_n)=\sum_{i=0}^{n} (-1)^i \rho(v_i)(\xi(v_0,\cdots,\hat{v_i},\cdots,v_n))\nonumber\\
		\qquad +\sum_{i<j}(-1)^{i+j}\xi([v_i,v_j],v_0,\cdots,\hat{v_i},\cdots,\hat{v_j},\cdots,v_n),
\end{eqnarray}
for any $v_0,\cdots,v_n\in \Gamma(A)$. It is easy to see $X_Q$ is of degree $1$. It can also be verified that $X_Q$ defined in this way satisfies $[[X_Q,X_Q]]=0$.
\begin{remark}
In literature, a super vector field of odd degree satisfying $[[X_Q,X_Q]]=0$ is called a homological vector field. For more details on homological vector fields and supermanifolds, see \cite{schatz} and \cite{vaintrob}.
\end{remark}

\par Locally we can write out the vector field $X_Q$ explicitly. Suppose $(\xi_1,\cdots,\xi_n)$ is a local frame of $A$ over some open chart $(U,x^1,\cdots,x^m)$ of $M$. The structure constants of the Lie algebroid is determined by
\begin{eqnarray}
	[\xi_i,\xi_j]=C_{ij}^k(x) \xi_k ,\\
	\rho(\xi)=b_i^t(x)\frac{\partial}{\partial x^t}.
\end{eqnarray}
Treat the dual coframe $(\xi^1,\cdots,\xi^n)$ as the coordinates on the fibers of $A|_U$, then
\begin{eqnarray}
\label{X_Qlocal}
	X_Q|_U = \frac{1}{2}C_{ij}^k(x) \xi^i\wedge \xi^j \wedge \frac{\partial}{\partial \xi^k}-b_i^t(x)\xi^i\frac{\partial}{\partial x^t}.
\end{eqnarray}
Here $\xi^i$ is of degree $1$, $\displaystyle \frac{\partial}{\partial \xi^i}$ is of degree $-1$, $\displaystyle \frac{\partial}{\partial x^t}$ is of degree $0$ and $C_{ij}^k(x), b_i^t(x)\in \mathcal{C}^{\infty}(M)$ are both of degree $0$.

\par It is well known that the Lie algebroid structure on $A$ and the $Q$-vector field $X_Q$ are equivalent.
\begin{proposition}
Given a vector bundle $A\to M$, there is a one-one correspondence between the Lie algebroid structures on $A$ and the homological vector fields of degree $1$ on $A[1]$. 
\end{proposition}

\par From the fact $[[X_Q,X_Q]]=0$, we also get
\begin{proposition}
\label{dgla}
$(\mathfrak{X}(A[1]), [[\cdot,\cdot]], [[X_Q,\cdot]])$ is a differential graded Lie algebra with differential $d= [[X_Q,\cdot]]$.
\end{proposition}

\par The deformation of the Lie algebroid structure over $A$ is equivalent to the deformation of the homological vector field $X_Q$ over $A[1]$. It is clear that $\tilde{X}_Q$ gives a deformation of $X_Q$, i.e. $X_Q+\tilde{X}_Q$ is again a homological vector field of degree $1$, if and only if 
\[
	[[X_Q+\tilde{X}_Q,X_Q+\tilde{X}_Q]]=0.
\]

\par Immediately we get the main theorem of this section:
\begin{theorem}
\label{thm1}
The deformation of the Lie algebroid $(A,[\cdot,\cdot],\rho)$ can be controlled by the DGLA in proposition \ref{dgla}. To be specific, $\tilde{X_Q}\in\mathfrak{X}(A[1])$ of degree $1$ makes $\tilde{X}_Q+X_Q$ a deformation of $X_Q$ if and only if it is a Maurer-Cartan element of the differential graded Lie algebra $(\mathfrak{X}(A[1],[[\cdot,\cdot]],[[X_Q,\cdot]])$, i.e. $\tilde{X}_Q$ satisfies the Maurer-Cartan equation
\[
	[[X_Q,\tilde{X}_Q]]+\frac{1}{2}[[\tilde{X}_Q,\tilde{X}_Q]]=0.
\]
\end{theorem}

\begin{remark}
\par M. Crainic and I. Moerdijk also find a DGLA to control the deformations of Lie algebroids(\cite{crainic}). It turns out the two DGLA's are isomorphic. Let's recall the construction in \cite{crainic} and show the equivalence.

\par Given the Lie algebroid $(A,[\cdot,\cdot],\rho)$, a derivation $D$ of degree $n$ is an antisymmetric $\mathbb{R}$-multi-linear map $\overbrace{\Gamma(A)\times\cdots\times\Gamma(A)}^{n+1\text{ times}}\to \Gamma(A)$ together with a map $\sigma_D:\overbrace{\Gamma(A)\times\cdots\times\Gamma(A)}^{n \text{ times}}\to \mathfrak{X}(M)$, called the symbol of $D$, such that
\[
	D(v_1,\cdots,v_n,f\cdot v_{n+1})=f\cdot D(v_1,\cdots,v_n)+(\sigma_{D}(v_1,\cdots,v_n)f)\cdot v_{n+1}
\]
for any $v_1,\cdots,v_{n+1}\in \Gamma(A)$ and $f\in \mathcal{C}^{\infty}(M)$.

\par Let $Der_n(A)$ be the set of all derivations of degree $n$, then $Der(A)=\sum_{n=0}^{\infty}Der_n(A)$ is a graded vector space. $Der(A)$ turns out to be a graded Lie algebra with respect to the bracket
\begin{eqnarray}
	\{D,E\}=-((-1)^{p\cdot q}D\circ E-E\circ D)
\end{eqnarray}
for any $D\in Der_p(A)$, $E\in Der_q(A)$.
\par It is easy to see the bracket $[\cdot,\cdot]$ of the Lie algebroid is a derivation of degree $1$ with symbol the anchor $\rho$. Denote this derivation by $D_Q$ and it can be checked $D_Q$ satisfies $\{D_Q,D_Q\}=0$. Thus $(Der(A),\{\cdot,\cdot\}, \{D_Q,\cdot\})$ is a DGLA. Given the local frame $(\xi_1,\cdots,\xi_n)$ and local coordinate $(x^1,\cdots,x^m)$ as before, for any $D=\dfrac{1}{(p+1)!}D_{i_1\cdots i_{p+1}}^a \xi^{i_1}\wedge\cdots\wedge \xi^{i_{p+1}}\otimes\xi_a\in Der_p(A)$ with symbol $\sigma_D=\dfrac{1}{p!}b_{j_1\cdots j_p}^t \xi^{j_1}\wedge\cdots\wedge\xi^{i_p} \wedge\frac{\partial}{\partial x^t}$, the map defined by
\[
	D\mapsto \frac{1}{(p+1)!}D_{i_1\cdots i_{p+1}}^a \xi^{i_1}\wedge\cdots\wedge \xi^{i_{p+1}}\wedge\frac{\partial}{\partial \xi_a} -\frac{1}{p!}b_{j_1\cdots j_p}^t \xi^{j_1}\wedge\cdots\wedge\xi^{i_p} \wedge\frac{\partial}{\partial x^t}
\]
is an isomorphism between the differential graded Lie algebra $(Der(A),\{\cdot,\cdot\},\{D_Q,\cdot\})$ and the differential graded Lie algebra $(\mathfrak{X}(A[1]),[[\cdot,\cdot]],[[X_Q,\cdot]])$ in proposition \ref{dgla}. 
\end{remark}

\begin{remark}
\label{differentialview}
There is another viewpoint of deformations of a Lie algebroid. A Lie algebroid structure ($[\cdot,\cdot]$ and $\rho$) on $A$ is equivalent to a differential $d$ on the graded commutative associative algebra $\mathcal{A}=\Gamma(\oplus_{n=0}^{\infty}\wedge^n A^*)$ (\cite{ping}) in the following way
\begin{eqnarray*}
	d(\xi)(v_0,\cdots,v_n)=\sum_{i=0}^{n} (-1)^i \rho(v_i)(\xi(v_0,\cdots,\hat{v_i},\cdots,v_n))\\
		\qquad +\sum_{i<j}(-1)^{i+j}\xi([v_i,v_j],v_0,\cdots,\hat{v_i},\cdots,\hat{v_j},\cdots,v_n),
\end{eqnarray*}
for any $\xi\in\Gamma(\wedge^n A^*)$ and $v_0,\cdots,v_n\in \Gamma(A)$. Here the differential $d$ is linear map $d:\mathcal{A}\to\mathcal{A}$ of degree $1$ satisfying $d^2=0$. Then a deformation of the Lie algebroid structure on $A$ is equivalent to a deformation of the differential $d$ of $\mathcal{A}$. It is easy to see a linear map $\tilde{d}:\mathcal{A}\to\mathcal{A}$ of degree $1$ deforms $d$ if and only if
\[
	d \circ \tilde{d}+\frac{1}{2}\tilde{d}\circ\tilde{d}=0.
\]
which is a Maurer-Cartan equation of the same form as in Theorem \ref{thm1}.
\end{remark}

\begin{remark}
When the Lie algebroid degenerates to a Lie algebra, our theorem \ref{thm1} recovers the classical result of A. Nijenhuis and R. W. Richardson (\cite{nijenhuis}).
\end{remark}

\section{Deformations of Lie subalgebroids}
\label{sec4}
\par In this section we will use Voronov's higher derived bracket to construct an $L_\infty$-algebra to control deformations of Lie subalgebroids. Suppose $(A,[\cdot,\cdot],\rho)$ is a Lie algebroid over $M$, and $E\subset A$ a Lie subalgebroid over a closed submanifold $S$. In order to define the deformation of the Lie subalgebroid $E$, we need some preparations.
\begin{itemize}
\item Find a complement subbundle $F$ over $S$ such that $A|_S=E\oplus F$.
\item Identify the normal bundle $NS=\dfrac{TM|_S}{TS}$ as a tubular neighborhood of $S$. As we are considering deformations not far from $E\to S$, without loss of generality, we assume $M=NS$.
\item The subbundle $E$ and $F$ can be pulled back along the projection $NS\to S$. Denote their pull-backs by $\tilde{E}$ and $\tilde{F}$ separately, i.e.
\[
	\begin{array}{ll}
	\begin{CD}
	\tilde{E} @>>> E  \\
	@VVV	     @VVV \\
	NS @>>> S 	
	\end{CD}
	\qquad
	\begin{CD}
	\tilde{F} @>>> F  \\
	@VVV	     @VVV \\
	NS @>>> S.
	\end{CD}
	\end{array}
\]
As $A|_S=(\tilde{E}\oplus \tilde{F})|_S$ and $S$ is closed, $\tilde{E}$ and $\tilde{F}$ still form a direct sum isomorphic to $A$ in an open neighborhood of $S$. Shrinking the tubular neighborhood $NS$ if necessary, we assume $A=\tilde{E}\oplus \tilde{F}$ over $M=NS$.
\end{itemize}

A deformation of the Lie subalgebroid $E$ should be a subbundle $E'\subset A$ over some submanifold $S'\subset M$ which is a again a Lie subalgebroid. As $E'$ is not far away from $E$, we only need to consider the case when $S'$ is a smooth section of $NS$ and $E'$ is the graph of some bundle map $\tilde{E}|_{S'}\to \tilde{F}|_{S'}$. In the following, denote the image of some $\sigma\in\Gamma(NS)$ by $S_{\sigma}$. We make the following definition.  

\begin{definition}
A deformation of the Lie subalgebroid $E\subset A$ is a pair $(\sigma,\phi)$ where $\sigma$ is a section of $NS$ and $\phi:\tilde{E}|_{S_{\sigma}}\to \tilde{F}|_{S_{\sigma}}$ is a bundle map such that the graph $gr(\phi)$ is a Lie subalgebroid of $A$, i.e.
\begin{itemize}
\item $\rho(gr(\phi))\subset TS_{\sigma}$,
\item For any $\xi,\eta \in \Gamma(gr(\phi))$ and any extension $\tilde{\xi},\tilde{\eta}\in \Gamma(A)$, the bracket satisfies $[\tilde{\xi},\tilde{\eta}]|_{S_{\sigma}}\in\Gamma(gr(\phi))$ and it does not depend on the choice of the extensions.
\end{itemize}
Any $(\sigma,\phi)\in \Gamma(NS\times (\tilde{E}^*|_{S_\sigma}\otimes \tilde{F}^*|_{S_\sigma}))$ is called a possible deformation.
\end{definition}

\par The following construction of $L_\infty$-algebra is motivated by the idea of strong homotopy Lie algebroid (\cite{schatz1}, \cite{cattaneo}). Denote the bundle projection $E\to S$ by $p$. Since $E$, $F$ and $NS$ are all vector bundles over $S$, we can pull back $F\to S$ and $NS\to S$ along $E\to S$ to get
\[
	\begin{CD}
		p^*F @>>> F\\
		@VVV	@VVV \\
		E @>p>> S
	\end{CD}\qquad
	\begin{CD}
		p^*NS @>>> NS\\
		@VVV	  @VVV\\
		E @>p>> S
	\end{CD}.
\]
Let $p^*F[1]$ be the degree $-1$ component, $p^*NS$ the degree $0$ component, then $p^*F[1]\oplus p^*NS$ becomes a super vector bundle over $E[1]$. Let $\mathfrak{a}=\Gamma(p^*F[1]\oplus p^*NS)$ be the collection of all sections of the super vector bundle.

\par The kernel of the bundle map
\[
	\begin{CD}
		TA @>>> T{\tilde{E}} @>>>  TE\\
		@VVV	@VVV		   @VVV\\
		A  @>>> \tilde{E}    @>>>  E.
	\end{CD}
\]
is the pull-back $(p^*)^3(F\oplus NS)$ of the bundle $p^*(F\oplus NS)\to E$ along $A\to \tilde{E} \to E$ by the diagram
\[
	\begin{CD}
		(p^*)^3(F\oplus NS)   @>>> (p^*)^2(F\oplus NS) @>>> p^*(F\oplus NS)\\	
		@VVV	@VVV	 @VVV \\
		A @>>> \tilde{E} @>>> E,
	\end{CD}
\]
hence we have an injection
\begin{eqnarray}
\label{injection}
	\Gamma(p^*(F\oplus NS))\to \Gamma((p^*)^3(F\oplus NS))\hookrightarrow \Gamma(TA). 
\end{eqnarray}
If letting $V=\mathfrak{X}(A[1])$ be the set of all super vector fields on $A[1]$, the map (\ref{injection}) induces an injective linear map $I:\mathfrak{a}\to V$.

\par On the other hand, for any vector bundle $H$ over $Q$, we have $TH|_Q=TQ\oplus H$ by the splitting short exact sequence
\[
\begin{CD}
	0 @>>> TQ @>i>> TH|_Q @>j>> H @>>> 0
\end{CD}
\]
where $i$ is induced from the inclusion of $Q$ as zero section into $H$ and $j$ is from the natural isomorphism $H=\frac{TH|_Q}{TQ}$. The sequence is splitting since there is a natural projection $TH|_Q\to TQ$ induced by the bundle projection $H\to Q$. Applying to the vector bundle $A\to E$, immediately we get
\[
	TA|_E=TE\oplus p^*(F\oplus NS),
\]
so there is also a projection $P:V\to \mathfrak{a}$ induced by
\begin{eqnarray}
	\Gamma(TA)\to \Gamma(TA|_E) \to \Gamma(p^*(F\oplus NS)).
\end{eqnarray}
It can be checked that $P\circ I=Id_\mathfrak{a}$, and hence $\mathfrak{a}$ can be treated as a direct summand of $V$. Moreover $I(\mathfrak{a})$ is an abelian subalgebra of $V$, and $ker(P)$ is closed under the bracket $[[\cdot,\cdot]]$, so the following result holds for our construction.

\begin{proposition}
\label{propv}
$(\mathfrak{X}(A[1]),[[\cdot,\cdot]],\mathfrak{a},P)$ froms a V-algebra.
\end{proposition}

\par In fact the bundle $p^*(F\oplus NS)\to E$ is isomorphic to the normal bundle $NE$ of $E$ in $A$ if considering $E$ as a submanifold of $A$. It is natural to identify the set of all possible deformations of the Lie subalgebroid $E$ with $\mathfrak{a}_0$. Denote the projection $\tilde{E}\to E$ by $p_E$ and $\tilde{F}\to F$ by $p_F$. Given any possible deformation $(\sigma,\phi)\in \Gamma(NS\oplus(\tilde{E}^*|_{S_\sigma}\otimes \tilde{F}|_{S_\sigma}))$, since $p_E:\tilde{E}|_{S_\sigma}\to E$ and $p_F: \tilde{F}|_{S_\sigma}\to F$ are isomorphisms, $\phi$ determines a map $\tilde{\phi}:E\to F$ by $p_F|_{\tilde{F}|_{S_\sigma}}\circ\phi\circ p_E^{-1}|_{\tilde{E}|_{S_\sigma}}$. Now $(\sigma,\tilde{\phi})$ is in $\Gamma(NS\oplus (E^*\otimes F))$ and thus can be pulled back to get a section $X_{\sigma,\phi}$ of $p^*NS\oplus p^*F[1]$ over $E[1]$. $X_{\sigma,\phi}$ is of degree $0$ and it is uniquely determined by $(\sigma,\phi)$.

\par Let us write out the explicit formulas for the correspondence $(\sigma,\phi)\mapsto X_{\sigma,\phi}$ in local coordinates. Let $\{x^1,\cdots,x^p\}$ be the local coordinates of an open subset $U\subset S$, $\{y_1,\cdots,y_q\}$ a local frame of $NS|_U$, $\{\xi_1,\cdots,\xi_l\}$ a local frame of $E|_U$ and $\{\xi_{l+1},\cdots,\xi_n\}$ a local frame of $F|_U$. Automatically the dual coframe $(y^1,\cdots,y^q)$, $\{\xi^1,\cdots,\xi^l\}$ and $\{\xi^{l+1},\cdots,\xi^n\}$ are coordinates of fibers of $NS|_U$, $E|_U$ and $F|_U$. Suppose the deformation $(\sigma,\phi)$ is given by
\begin{eqnarray}
	\sigma:S\to NS \quad (x^1,\cdots,x^p)\mapsto (\sigma^1(x),\cdots,\sigma^q(x))\\
	\phi: \tilde{E}|_{S_\sigma}\to \tilde{F}|_{S_\sigma} \quad (\sigma(x),\xi^i) \mapsto \sum_{j=l+1}^{n}(\sigma(x),\phi_i^j(x) \xi^j).
\end{eqnarray}
Then the section $X_{\phi,\sigma}\in \mathfrak{a}$ determined by $(\sigma,\phi)$ equals
\begin{eqnarray}
	X_{\sigma,\phi} = \sum_{i=1}^{l}\sum_{j=l+1}^n\phi_i^j(x)\xi^i\wedge \xi_j + \sigma^k(x)y_k,
\end{eqnarray}
which is of degree $0$. As $\tilde{E}$ and $\tilde{F}$ are pull-backs of $E$ and $F$, the sections $\xi_i$'s can be pulled back to independent sections $\tilde{\xi}_i$'s of $\tilde{E}$ and $\tilde{F}$ over $\pi_{NS}^{-1}(U)$ where $\pi_{NS}:NS\to S$ is the projection. Similar for the coframes. The image $I(X)$ of $X_{\sigma,\phi}$ under the map $I: \mathfrak{a}\to V$ looks like
\begin{eqnarray}
	I(X_{\sigma,\phi}) = \sum_{i=1}^{l}\sum_{j=l+1}^n\phi_i^j(x)\tilde{\xi}^i\wedge\frac{\partial}{\partial \tilde{\xi}^j} + \sigma^k(x)\frac{\partial}{\partial y^k}.
\end{eqnarray}

\par From the local expressions, it is clear that the map $(\sigma,\phi)\mapsto X_{\sigma,\phi}$ is a 1-1 correspondence between the possible deformations of the Lie subalgebroid $E$ and the elements in $\mathfrak{a}_0$. We will identify $\mathfrak{a}_0$ as the collection of all possible deformations of the Lie subalgebroids and abuse the notation to write $(\sigma,\phi)\in\mathfrak{a}_0$ to indicate it is a possible deformation of $E$. 

\par We are going to show under certain conditions, this $L_\infty[1]$-algebra controls the deformation of the Lie subalgebroid $E$. Geometrically $E\subset A$ is a Lie subalgebroid if and only if the homological vector field $X_Q$ is tangent to the submanifold $E[1]$. This can be seen in the following way. $X_Q$ acting on smooth function $f\in\mathcal{C}^\infty(A[1])$ is the same as the differential $d$ acting on $f\in\Gamma(\wedge^* A^*)$ (see \ref{X_Qdef} and Remark \ref{differentialview}). It is well known that $E$ is a Lie subalgebroid if and only if $d$ can restrict to a differential over $\Gamma(\wedge^* E^*)$, but this is equivalent to the condition that $X_Q$ is tangent to $E[1]$. To describe the conditions for $X_Q$ tangent to $E[1]$, we need the following lemma.

\par Suppose $Q$ is a manifold and $R\subset Q$ a submanifold. Identify the normal bundle $NR$ as a tubular neighborhood of $R$ in $Q$. As before, a deformation of $R$ is defined to be the graph of some $\gamma\in\Gamma(NR)$. Without loss of generality identify $NR$ with $Q$. Similarly $TQ|_R=TR\oplus NR$ and the kernel of the projection $TQ\to TR$ induced by the projection $Q=NR\to R$ is the pull-back $p^*(NR)$ of $NR\to R$ along $NR\to R$, so any $\gamma\in\Gamma(NR)$ corresponds to a vector field $Y_\gamma\in\mathfrak{X}(Q)$ by
\[
	\Gamma(NR)\to \Gamma(p^*(NR))\to \mathfrak{X}(Q),
\]
and there is also a projection $Pr: \mathfrak{X}(Q)\to \Gamma(TQ|_R) \to \Gamma(NR)$. 
\par Here for simplicity we make the conversion that a function $f\in\mathcal{C}^{\infty}(Q)$ is analytic along the fibers of $NR$ if for any $x\in R$, $f|_{N_xR}$ has a Taylor expansion at $(x,0)$ convergent along the whole fiber $N_xR$. A vector field $Z\in\mathfrak{X}(Q)$ is analytic along the fibers of $NR$ if the coefficients of $Z$ under any smooth local frame of $TQ$ is analytic along the fibers of $NR$.

\begin{lemma}
\label{lem1}
Suppose $Z\in\mathfrak{X}(Q)$ is a vector field tangent to $R$ and analytic along the fibers of $NR$. Given a section $\gamma\in\Gamma(NR)$, $Z$ is tangent to the graph of $\gamma$ if and only if
\begin{eqnarray}
	\sum_{n=1}^{\infty}\frac{1}{n!} Pr([\cdots[Z,\overbrace{-Y_\gamma],\cdots,-Y_\gamma]}^{n\text{ times}})=0,
\end{eqnarray}
\end{lemma}
\begin{proof}
\par Working locally, let $(z^i)$ be the local coordinate of $R$ and $(y^j)$ the local coordinates of the fibers of $NR$. $Z$ is of the form
\begin{eqnarray}
	Z=Z^i(z,y)\frac{\partial}{\partial z_i}+ Y^j(z,y)\frac{\partial}{\partial y^j},
\end{eqnarray}
satisfying $Y^j|_R=0$ or $Y^j(z,0)=0$ for all $j$. Suppose the section $\gamma:R\to NR$ is defined by $z\mapsto y$ where $y^j=\gamma^j(z)$, then locally the vector field $Y_\gamma$ is given by
\begin{eqnarray}
	Y_\gamma=\gamma^j(z)\frac{\partial}{\partial y^j}.
\end{eqnarray}
$Z$ is tangent to the graph $gr(\gamma)$ if and only if $Y^j(z,\gamma(z))\frac{\partial}{\partial y^j}=\gamma_* (Z^i(z,\gamma(z))\frac{\partial}{\partial z^i})$ or
\begin{eqnarray}
\label{taylor0}
	Y^j(z,\gamma(z))\frac{\partial}{\partial y^j}=Z^i(z,\gamma(z))\frac{\partial \gamma^j}{\partial x^i}\frac{\partial}{\partial y^j}
\end{eqnarray}
Since $Z$ is analytic along the fibers of $NR$, both $Z^i(z,y)$ and $Y^j(z,y)$ are analytic in $y$ at $(z,0)$ whose Taylor expansions are convergent for all $y$. If we write the equation (\ref{taylor0}) into Taylor expansion with respect to $y$-coordinates, we get
\begin{eqnarray}
	\sum_{n=0}^{\infty}\frac{1}{n!}\sum_{|K|=n}\frac{\partial^K Y^j(z,0)}{\partial y^K}(\gamma^j(z))^K=\sum_{n=0}^{\infty}\frac{1}{n!}\frac{\partial \gamma^j(z)}{\partial z^i} \sum_{|K|=n}\frac{\partial^K X^i(z,0)}{\partial y^K}(\gamma^j(z))^K.
\end{eqnarray}
Here $K$ stands for multi-indices, and if $K=(k_1,\cdots,k_l)$ then $|K|=\sum_{i=1}^{l}k_i$. It turns out this equation is identical with $\sum_{n=1}^{\infty}\frac{1}{n!} Pr([\cdots[Z,\overbrace{-Y_\gamma],\cdots,-Y_\gamma]}^{n\text{ times}})=0$ in local coordinates. The equation does not depend on the local coordinates and we proved our lemma. 
\end{proof}

\par Since $E$ is a Lie subalgebroid, $X_Q$ is tangent to $E[1]$, as a result $P(X_Q)=0$. Besides $X_Q$ is of degree $1$ satisfying $[[X_Q,X_Q]]=0$, so we can apply Theorem \ref{v1}. Immediately we get the first part of the following theorem.

\begin{theorem}
\label{thm2}
$\mathfrak{a}_{X_Q}^P$ is a flat $L_\infty[1]$-algebra with the structure maps
\begin{eqnarray}
	m_k(a_1,\cdots,a_k)=P[[\cdots[[X_Q,I(a_1)]],\cdots,I(a_k)]]
\end{eqnarray} 
for all $k\ge 1$. If the Lie algebroid structure $X_Q$ is analytic along the fibers of $NS$, $(\sigma,\phi)\in \Gamma(NS\oplus(\tilde{E}^*|_{S_\sigma}\otimes \tilde{F}|_{S_\sigma}))$ is a deformation of the Lie subalgebroid $E$ if and if $-X_{\sigma,\phi}$ is a Maurer-Cartan element of the $L_\infty[1]$-algebra $\mathfrak{a}_{X_Q}^P$.
\end{theorem}

\begin{remark}
Recall that the local expression of $X_Q$ is of the form
\begin{eqnarray}
	X_Q|_U = \frac{1}{2}C_{ij}^k(x,y) \tilde{\xi}^i\wedge \tilde{\xi}^j \wedge \frac{\partial}{\partial \tilde{\xi}^k}-b_i^t(x,y)\tilde{\xi}^i\frac{\partial}{\partial x^t}
\end{eqnarray}
which is already analytic (in fact polynomial) along the fibers of $E$ and $F$(over $S$). The Lie algebroid structure is analytic along the fibers of $NS$ is equivalent to the structure constants $C_{ij}^k(x,y)$ and $b_i^t(x,y)$ are analytic in $y$ at $(x,0)$ and convergent for all $y$.
\end{remark}

\begin{proof}
$(\sigma,\phi)\in \Gamma(NS\oplus(\tilde{E}^*|_{S_\sigma}\otimes \tilde{F}|_{S_\sigma}))$ is a deformation of the Lie subalgebroid $E$ if and only if $X_Q$ is tangent to the graph of $\phi$. $X_Q$ is already analytic along the fiber of $F$, and hence analytic along the fibers of $p^*F[1]$ over $E[1]$. By assumption it is also analytic along the fibers of $NS$, and hence analytic along the fiber of $p^*NS$ over $E[1]$. As $p^*(NS\oplus F)$ is the tubular neighborhood of $E$ in $A$, by Lemma \ref{lem1}, $X_Q$ is tangent to $gr(\phi)$ if and only if the vector field $I(X_{\sigma,\phi})$ corresponding to $(\sigma,\phi)$ satisfies the following equation
\begin{eqnarray}
	\sum_{n=1}^{\infty}\frac{1}{n!}[[\cdots,[[X_Q, -I(X_{\sigma,\phi})]],\cdots,-I(X_{\sigma,\phi})]]=0.
\end{eqnarray}
This is the same as the Maurer-Cartan equation of $\mathfrak{a}_{X_Q}^P$, which means $-X_{\sigma,\phi}$ is a Maurer-Cartan element of $\mathfrak{a}_{X_Q}^P$.
\end{proof}

\par If we fix $\sigma=0$, then a deformation $(\sigma,\phi)$ of $E$ is simply a bundle map $\phi:E\to F$. Let $\mathfrak{a}'=\Gamma(p^*F[1])$ where $p^*F$ is the pull-back of $F\to S$ over $E\to S$. Then $\mathfrak{a}'_0$ becomes the subset of $\mathfrak{a}_0$ consisting of possible deformations of the form $(0,\phi)$. We can identify $\mathfrak{a}'_0$ with the collection of all possible deformations over the fixed base $S$. Since $\mathfrak{a}'\subset\mathfrak{a}$, automatically $I(\mathfrak{a}')$ is an abelian subalgebra of $V$. Moreover the projection $p^*(F\oplus NS)\to p^*(F)$ induces a projection $P_1: \mathfrak{a}\to\mathfrak{a'}$ and the map $\tilde{P}:V\to \mathfrak{a}\to\mathfrak{a}'$ satisfying $\tilde{P}\circ I=Id_{\mathfrak{a}'}$. It also can be checked that $ker(\tilde{P})$ is a Lie subalgebra of $V$.
\begin{corollary}
\label{cor1}
$(V,[[\cdot,\cdot]],\mathfrak{a}',\tilde{P})$ is a V-algebra and ${\mathfrak{a}'}_{X_Q}^{\tilde{P}}$ is an $L_\infty[1]$-algebra. The graph of $\phi:E\to F$ is a deformation of the Lie subalgebroid if and only if $-X_{0,\phi}$ is a Maurer-Cartan element of this $L_\infty[1]$-algebra.
\end{corollary}

\par One thing to mention is that as we only consider the deformation of $E$ over $S$, the behavior of the Lie algebroid structure $X_Q$ outside the zero section of $NS$ does not matter here, so we do not require $X_Q$ being analytic along the fibers of $NS$ any more. The corollary follows from our Theorem \ref{thm2} immediately.

\par Another thing to mention is that only the first $3$ structure maps $m_1$, $m_2$, $m_3$ are nonzero and $m_k=0$ for $k\ge 4$ in this flat $L_\infty[1]$-algebra. We can also give the explicit formulas for these structure maps. Denote the projection $A \to E$ by $\pi_{E}$ and $A\to F$ by $\pi_{F}$. As vector spaces, $\mathfrak{a}'=\Gamma(p^*F[1])$ is isomorphic to $\Gamma(\sum_{n=0}^{\infty}\wedge^n E^*\otimes F)[1]$ by identifying the tangent space $T(F_s)$ at the point $s\in S$ with $F_{s}$. Let the degree of the component $\Gamma(\wedge^n E^*\otimes F)$ to be $n$. As $\tilde{E}$ and $\tilde{F}$ are pull-backs of $E$ and $F$ along $NS\to S$, any $\xi\in \Gamma(\wedge^*E^*\otimes F)$ or $a\in E$ can be extended to a unique section $\tilde{\xi}\in\Gamma(\tilde{E}^*\otimes \tilde{F})$ or $\tilde{a}\in \Gamma(\tilde{E})$ by constantly pulling back along the fibers of $NS$. The structure maps of the $L_\infty[1]$-algebra $\mathfrak{a}'\cong\Gamma(\sum_{n=0}^{\infty}\wedge^n E^*\otimes F)[1]$ in \ref{cor1} are given by
\begin{eqnarray*}
	m_1(\xi[1])(a_1,\cdots,a_{k+1}) = \sum_{\tau\in S_{k,1}} (-1)^{\tau} \pi_{F}([\tilde{\xi}(\tilde{a}_{\tau(1)},\cdots,\tilde{a}_{\tau(k)}),\tilde{a}_{\tau(k+1)}])|_S\\
         \quad -(-1)^{k-1} \sum_{\tau\in S_{2,k-2}} (-1)^{\tau} \xi(\pi_{E}([\tilde{a}_{\tau(1)},\tilde{a}_{\tau(2)}]|_S),a_{\tau(3)},\cdots,a_{\tau(k+1)})
\end{eqnarray*}
\begin{eqnarray*}
	m_2(\xi[1],\psi[1])(a_1,\cdots,a_{k+l}) =  (-1)^k\sum_{\tau\in S_{l,k}} (-1)^{\tau} \pi_{F}([\tilde{\xi}(\tilde{a}_{\tau(l+1)},\cdots,\tilde{a}_{\tau(l+k)}),\tilde{\psi}(\tilde{a}_{\tau(1)},\cdots, \tilde{a}_{\tau(l)})])|_S\\
- (-1)^k\sum_{\tau\in S_{l,1,k-1}} (-1)^{\tau} \xi (\pi_{E}([\tilde{\psi}(\tilde{a}_{\tau(1)},\cdots,\tilde{a}_{\tau(l)}),\tilde{a}_{\tau(l+1)}])|_S,a_{\tau(l+2)},\cdots,a_{\tau(k+l)})\\
 + (-1)^{(l-1)k}\sum_{\tau\in S_{k,1,l-1}} (-1)^{\tau} \psi(\pi_{E}([\tilde{\xi}(\tilde{a}_{\tau(1)},\cdots,\tilde{a}_{\tau(k)}),\tilde{a}_{\tau(k+1)}])|_S,a_{\tau(k+2)},\cdots,a_{\tau(k+l)})
\end{eqnarray*}
\vspace{1mm}
\[
m_3(\xi[1],\psi[1],\phi[1])(a_1,\cdots,a_{k+l+m-1})=
\]
\vspace{-5mm}
\begin{eqnarray*}
\pm  \sum_{\tau\in S_{k,l,m-1}} (-1)^{\tau} \phi(\pi_{E}([\tilde{\xi}(\tilde{a}_{\tau(1)},\cdots, \tilde{a}_{\tau(k)}),\tilde{\psi}(\tilde{a}_{\tau(k+1)},\cdots,\tilde{a}_{\tau(k+l)})])|_S,a_{\tau(k+l+1)},\cdots,a_{\tau(k+l+m-1)})\\
\pm  \sum_{\tau\in S_{l,m,k-1}} (-1)^{\tau} \xi(\pi_{E}([\tilde{\psi}(\tilde{a}_{\tau(1)},\cdots,\tilde{a}_{\tau(l)}),\tilde{\phi}(\tilde{a}_{\tau(l+1)},\cdots, \tilde{a}_{\tau(l+m)})])|_S,a_{\tau(l+m+1)},\cdots,a_{\tau(k+l+m-1)})\\
\pm  \sum_{\tau\in S_{m,k,l-1}} (-1)^{\tau} \psi(\pi_{E}([\tilde{\phi}(\tilde{a}_{\tau(1)},\cdots,\tilde{a}_{\tau(m)}),\tilde{\xi}(\tilde{a}_{\tau(m+1)},\cdots, \tilde{a}_{\tau(m+k)})])|_S,a_{\tau(k+m+1)},\cdots,a_{\tau(k+l+m-1)})\\
\end{eqnarray*}
for any $\xi\in\Gamma(\wedge^k E^*\otimes F)$, $\psi\in\Gamma(\wedge^l E^*\otimes F)$, $\phi\in\Gamma(\wedge^m E^*\otimes F)$ and $a_i\in \Gamma(E)$.

\par If $S=M$ and consider the deformation of $E$, the $L_\infty[1]$-algebra ${\mathfrak{a}'}_{X_Q}^{\tilde{P}}$ is even simpler. In this case $A=E\oplus F$. Our Corollary \ref{cor1} can be rewritten as the following

\begin{corollary} 
\label{cor4}
If $E$ is a Lie subalgebroid over $M$, there is a flat $L_\infty[1]$-algebra structure on $\Gamma(\sum_{n=0}^{\infty}\wedge^n E^*\otimes F)[1]$ with structure maps given by
\begin{eqnarray*}
	m_1(\xi[1])(a_1,\cdots,a_{k+1}) = \sum_{\tau\in S_{k,1}} (-1)^{\tau} \pi_{F}([\xi(a_{\tau(1)},\cdots,a_{\tau(k)}),a_{\tau(k+1)}])\\
         \quad -(-1)^{k-1} \sum_{\tau\in S_{2,k-2}} (-1)^{\tau} \xi(\pi_{E}[a_{\tau(1)},a_{\tau(2)}],a_{\tau(3)},\cdots,a_{\tau(k+1)})
\end{eqnarray*}
\begin{eqnarray*}
	m_2(\xi[1],\psi[1])(a_1,\cdots,a_{k+l}) =  (-1)^k\sum_{\tau\in S_{l,k}} (-1)^{\tau} \pi_{F}([\xi(a_{\tau(l+1)},\cdots,a_{\tau(l+k)}),\psi(a_{\tau(1)},\cdots, a_{\tau(l)})])\\
- (-1)^k\sum_{\tau\in S_{l,1,k-1}} (-1)^{\tau} \xi (\pi_{E}([\psi(a_{\tau(1)},\cdots,a_{\tau(l)}),a_{\tau(l+1)}]),a_{\tau(l+2)},\cdots,a_{\tau(k+l)})\\
 + (-1)^{(l-1)k}\sum_{\tau\in S_{k,1,l-1}} (-1)^{\tau} \psi(\pi_{E}([\xi(a_{\tau(1)},\cdots,a_{\tau(k)}),a_{\tau(k+1)}]),a_{\tau(k+2)},\cdots,a_{\tau(k+l)})
\end{eqnarray*}
\vspace{1mm}
\[
m_3(\xi[1],\psi[1],\phi[1])(a_1,\cdots,a_{k+l+m-1})=
\]
\vspace{-5mm}
\begin{eqnarray*}
\pm  \sum_{\tau\in S_{k,l,m-1}} (-1)^{\tau} \phi(\pi_{E}([\xi(a_{\tau(1)},\cdots, a_{\tau(k)}),\psi(a_{\tau(k+1)},\cdots,a_{\tau(k+l)})]),a_{\tau(k+l+1)},\cdots,a_{\tau(k+l+m-1)})\\
\pm  \sum_{\tau\in S_{l,m,k-1}} (-1)^{\tau} \xi(\pi_{E}([\psi(a_{\tau(1)},\cdots,a_{\tau(l)}),\phi(a_{\tau(l+1)},\cdots, a_{\tau(l+m)})]),a_{\tau(l+m+1)},\cdots,a_{\tau(k+l+m-1)})\\
\pm  \sum_{\tau\in S_{m,k,l-1}} (-1)^{\tau} \psi(\pi_{E}([\phi(a_{\tau(1)},\cdots,a_{\tau(m)}),\xi(a_{\tau(m+1)},\cdots, a_{\tau(m+k)})]),a_{\tau(k+m+1)},\cdots,a_{\tau(k+l+m-1)})\\
\end{eqnarray*}
and $m_i=0$ for $i\ge 4$. This $L_\infty[1]$-algebra governs the deformation of the Lie subalgebroid $E$, i.e. if $\phi\in \Gamma(E^*\otimes F)$ such that $gr(\phi)$ is a Lie subalgebroid of $A$ if and only if $\phi[1]$ is a Maurer-Cartan element of this $L_\infty[1]$-algebra.
\end{corollary}

\begin{remark}
For the special case described in Corollary \ref{cor4}, as $M=S$, $NS$ is trivial and $\mathfrak{a}=\mathfrak{a}'=\mathfrak{X}(p^*(F))$. The non-involvement of $p^*(NS)$ means the anchor condition is automatically satisfied and this is obviously true since $\rho(gr(\phi))\subset TS=TM$ always holds for all deformations $\phi:E\to F$.
\end{remark}

\begin{remark}
If the Lie algebroid $A$ degenerates to a Lie algebra, then $E$ degenerates to a Lie subalgebra and our Corollary \ref{cor1} and \ref{cor4} recover the result of deformations of Lie subalgebras in \cite{zambon}.
\end{remark}

\begin{remark}
\label{complex}
This corollary also applies to the deformations of complex Lie subalgebroid (see \cite{ping2} for definitions). As $S=M$, the step of finding tubular neighborhoods is skipped. All other constructions can be carried out parallely to complex bundles, so the complex version of Corollary \ref{cor4} also works.
\end{remark}

\begin{remark}
The $L_\infty$-algebra controlling the deformation of the Lie subalgebroid $E$ can also be constructed from the viewpoint of Poisson manifold. Let us show the idea briefly. The dual $A^*$ of the Lie algebroid $A$ is a Poisson manifold(\cite{weinstein}). Using the same local coordinates as before, the Poisson bivector field is of the form
\[
	\Pi|_U = \frac{1}{2}C_{ij}^k(x,y) \tilde{\xi}_k \frac{\partial}{\partial \tilde{\xi}_i} \wedge \frac{\partial}{\partial \tilde{\xi}_j}-b_i^t(x,y)\frac{\partial}{\partial \tilde{\xi}_i}\wedge\frac{\partial}{\partial x^t}.
\]
$E$ is a Lie subalgebroid is equivalent to $E^{\bot}\cong F^*$ is a coisotropic submanifold of $A^*$. (Here $E^{\bot}$ is defined to be the set $\bigcup_{p\in S}\{a^*\in A^*_p|a^*(e)=0 \quad \forall e\in E_p\}$) A.S.Cattaneo and G.Felder showed in \cite{cattaneo} that 
\begin{quote}
(*) given a Poisson manifold $(Q,\Pi)$ and a coisotropic submanifold $R$, there is an $L_\infty$-algebra structure on $\Gamma(\sum_{n=0}^{\infty}\wedge^n NR)$.
\end{quote}
F.Sch$\ddot{a}$tz and M.Zambon showed in \cite{zambon1} that 
\begin{quote}
(**) if the Poisson bivector field $\Pi$ is analytic along the fiber of $NR$, for any $\gamma\in\Gamma(NR)$, the graph $gr(\gamma)$ is a deformation of the coisotropic submanifold if and only if $\gamma$ is a Maurer-Cartan element of the $L_\infty$-algebra $\Gamma(\sum_{n=0}^{\infty}\wedge^n NR)$.
\end{quote}
In our case, $F^*$ is the coisotropic submanifold and it is not hard to see the normal bundle of $NF^*$ in $A^*$ is isomorphic to the pull-back of $E^*\oplus NS$ along $F\to S$, i.e.
\[
\begin{CD}
	pr^*(E^*\oplus NS)	@>>>	E^*\oplus NS	\\
	@VVV 		@VVV		\\
	F^* 	@>pr>> 	S.
\end{CD}
\]
By (*) $\Gamma (\sum_{n=0}^{\infty}\wedge^n pr^*(E^*\oplus NS))$ is an $L_\infty$-algebra. $(\sigma,\phi)\in \Gamma(NS\oplus (\tilde{E}^*|_{S_{\sigma}}\otimes \tilde{F}|_{S_{\sigma}}))$ is a deformation of the Lie subalgebroid $E$ if and only if $(gr(\phi))^{\bot}=gr(-{\phi}^*)$ is a coisotropic submanifold of $A^*$. Notice that $\phi^*:\tilde{F}^*|_{S_{\sigma}}\to \tilde{E}^*|_{S_{\sigma}}$ corresponds to a unique $Z_{\sigma,\phi}\in\Gamma(pr^*(E^*\oplus NS))$ by pull-backs similar to $(\sigma,\phi)\mapsto X_{\sigma,\phi}$, so $(\sigma,\phi)\mapsto -Z_{\sigma,\phi}$ defines an injective map from the set of the deformations of the Lie subalgebroid $E$ to $\Gamma(pr^*(E^*\oplus NS))$. Denote the image of this injective map by $\mathfrak{b}$. It can be checked $\mathfrak{b}$ is an $L_\infty$-subalgebra of $\Gamma (\sum_{n=0}^{\infty}\wedge^n pr^*(E^*\oplus NS))$. If we assume the Lie algebroid structure $X_Q$ is analytic along the fibers of $NS$, then the Poisson bivector $\Pi$ is analytic along the fibers of $NF^*$ and hence by (**) $(\gamma,\phi)$ gives a deformation of $E$ as Lie subalgebroids if and only if $-Z_{\sigma,\phi}$ is a Maurer-Cartan element of $\mathfrak{b}$.
\end{remark}

\section{Simultaneous deformations of Lie algebroids and their Lie subalgebroids}
\label{sec5}
\par By Theorem \ref{thm1} we have a DGLA $(\mathfrak{X}(A[1]),[[\cdot,\cdot]],[[X_Q,\cdot]])$ to control deformations of the Lie algebroid structure on $A$, and by Theorem \ref{thm2} we have an $L_\infty[1]$-algebra $\mathfrak{a}_{X_Q}^P$ to control the deformation of the Lie subalgebroid $E$. In this section we are going to combine the two structures to control the simultaneous deformation of the Lie algebroid $A$ and its Lie subalgebroid $E$. The tool we are going to use is Theorem \ref{v2} and Theorem \ref{zam1}.

\par By Propostion \ref{propv} $(\mathfrak{X}(A[1]),[[\cdot,\cdot]],\mathfrak{a},P)$ forms a V-algebra. By Theorem \ref{v2}, 
\[
	(\mathfrak{X}(A[1])[1]\oplus\mathfrak{a})_{X_Q}^P
\] 
is an $L_\infty[1]$-algebra. Then by Theorem \ref{zam1} and Remark \ref{convergence} we have
\begin{corollary}
\label{cor2}
If the homological vector field $X_Q$ is analytic along the fibers of $NS$ and $(\sigma,\phi)\in \Gamma(NS\oplus (\tilde{E}^*|_{S_\sigma}\otimes \tilde{F}|_{S_\sigma}))$ gives a deformation of the Lie subalgebroid $E$, then $\mathfrak{a}_{X_Q}^{P_{-I(X_{\sigma,\phi})}}$ is a flat $L_\infty[1]$-algebra such that for any $(\tilde{\sigma},\tilde{\phi})$,
\begin{eqnarray*}
	X_{\tilde{\sigma},\tilde{\phi}}\text{ is a MC-element of } \mathfrak{a}_{X_Q}^{P_{-I(X_{\sigma,\phi})}}\\
\qquad \Leftrightarrow (\sigma+\tilde{\sigma},\phi+\tilde{\phi}) \text{ is a new deformation of } E.
\end{eqnarray*}
Moreover, for any $\tilde{X}_Q\in ker(P)_1$ and any $(\tilde{\sigma},\tilde{\phi})\in \Gamma(NS\oplus (\tilde{E}^*|_{S_{\tilde{\sigma}}}\otimes \tilde{F}|_{S_{\tilde{\sigma}}}))$,
\[
(\tilde{X}_Q,X_{\tilde{\sigma},\tilde{\phi}})\text{ is a MC-element of }(\mathfrak{X}(A[1])[1]\oplus\mathfrak{a})_{X_Q}^{P_{-I(X_{\sigma,\phi})}}\Leftrightarrow
\]
\vspace{-0.5cm}
\[
				\left\{\begin{array}{l}
				X_Q+\tilde{X}_Q \text{ defines a Lie algebroid structure on } A\\
				(\sigma+\tilde{\sigma},\phi+\tilde{\phi})\text{ gives a Lie subalgebroid w.r.t. the new Lie algebroid structure.}
				\end{array}\right.
\]
\end{corollary}

\par Here $P_{-I(X_{\sigma,\phi})}=P\circ e^{[[\cdot,-I(X_{\sigma,\phi})]]}$(see \ref{exponents}). At this point we get an $L_\infty[1]$-algebra governing the simultaneous deformation the Lie algebroid $A$ and its Lie subalgebroid $E$. If we take $(\sigma,\phi)=(0,0)$, i.e. the trivial deformation the Lie subalgebroid $E$ ($E$ itself), then the Corollary \ref{cor2} reads
\begin{corollary}
If the homological vector field $X_Q$ is analytic along the fibers of $NS$, for any $\tilde{X}_Q\in ker(P)_1$ and any $(\tilde{\sigma},\tilde{\phi})\in \Gamma(NS\oplus (\tilde{E}^*|_{S_{\tilde{\sigma}}}\otimes \tilde{F}|_{S_{\tilde{\sigma}}}))$,
\[
(\tilde{X}_Q,X_{\tilde{\sigma},\tilde{\phi}})\text{ is a MC-element of }(\mathfrak{X}(A[1])[1]\oplus\mathfrak{a})_{X_Q}^P\Leftrightarrow
\]
\vspace{-0.5cm}
\[
				\left\{\begin{array}{l}
				X_Q+\tilde{X}_Q \text{ defines a Lie algebroid structure on } A\\
				(\tilde{\sigma},\tilde{\phi})\text{ gives a Lie subalgebroid w.r.t. the new Lie algebroid structure.}
				\end{array}\right.
\]
\end{corollary} 

\par If we fix the base $S$ and consider the deformations of $E$ over $S$ only, by Corollary \ref{cor1} and Theorem \ref{zam1} we get another special case of Corollary \ref{cor2}. This time $(\mathfrak{X}(A[1]),[[\cdot,\cdot]],\mathfrak{a}',\tilde{P})$ is our V-algebra and $\tilde{P}_{-I(X_{0,\phi})}=P_1\circ P_{-I(X_{0,\phi})}$.

\begin{corollary}
\label{cor3}
If $\phi\in \Gamma(E^*\otimes F)$ gives a deformation of the Lie subalgebroid $E$ over $S$, then ${\mathfrak{a}'}_{X_Q}^{P_{-I(X_{0,\phi})}}$ is a flat $L_\infty[1]$-algebra such that for any $\tilde{\phi}\in \Gamma(E^*\otimes F)$
\begin{eqnarray*}
	X_{0,\phi}\text{ is a MC-element of } {\mathfrak{a}'}_{X_Q}^{\tilde{P}_{-I(X_{0,\phi})}}\Leftrightarrow \phi+\tilde{\phi} \text{ is a new deformation of } E\text{ over } S.
\end{eqnarray*}
Moreover, for any $\tilde{X}_Q\in ker(P)_1$ and any $\tilde{\phi}\in \Gamma(E^*\otimes F)$,
\[
(\tilde{X}_Q,X_{0,\tilde{\phi}})\text{ is a MC-element of }(\mathfrak{X}(A[1])[1]\oplus\mathfrak{a}')_{X_Q}^{\tilde{P}_{-I(X_{0,\phi})}}\Leftrightarrow
\]

\vspace{-0.5cm}
\[
				\left\{\begin{array}{l}
				X_Q+\tilde{X}_Q \text{ defines a new Lie algebroid structure on } A\\
				\phi+\tilde{\phi}\text{ gives a Lie subalgebroid w.r.t. the new Lie algebroid structure over }S.
				\end{array}\right.
\]
\end{corollary}

\begin{remark}
If the Lie algebroid $A$ degenerates to a Lie algebra and the Lie subalgebroid $E$ degenerates to a Lie subalgebra, our Corollary \ref{cor3} recovers the result of Y. Fr$\acute{e}$gier and M. Zambon. The $L_\infty[1]$-algebra controlling the simultaneous deformation of the Lie algebra and the Lie subalgebra in \cite{zambon} becomes a special case of our result.
\end{remark}

\section{Applications}
\label{sec6}
\subsection{Deformation of Foliations}
\par Consider the smooth manifold $M$ and a foliation $\mathcal{F}$ over it. According to Frobenius' Theorem, the foliation $\mathcal{F}$ is equivalent to a unique integrable distribution $D$, i.e. a subbundle $D\subset TM$ satisfying $[\Gamma(D),\Gamma(D)]\subset\Gamma(D)$. Here the bracket $[\cdot,\cdot]$ is the usual Schouten-Nijenhuis bracket of vector fields.

\par A deformation of $\mathcal{F}$ is just the deformation of $D$ as integrable distributions. First we choose a subbundle $F\subset TM$ such that $TM=D\oplus F$. Define the deformation of the integrable distribution $D$ to be the graph of some bundle map $\Phi:D\to F$ which is also integrable.

\par Notice that $(TM,[\cdot,\cdot],Id_{TM})$ is a Lie algebroid and a subbundle $E\subset TM$ over $M$ is a Lie subalgebroid if and only if $E$ is integrable, thus deformations of the integrable distribution $D$ are the same as deformations of the Lie subalgebroid $D$. In this case, the base manifold $S$ of our Lie subalgebroid is $M$ and we can apply Corollary \ref{cor4}.

\begin{proposition}
There is a flat $L_\infty[1]$-algebra structure on $\mathbb{L}=\Gamma(\sum_{n=0}^{\infty}\wedge^*D\otimes F)[1]$ with the structure maps the same as in Corollary \ref{cor4}. The graph $gr(\Phi)$ of some $\Phi\in \Gamma(D^*\otimes F)$ is a deformation of the distribution $D$ if and only if $\Phi[1]$ is a Maurer-Cartan element of this $L_\infty[1]$-algebra.
\end{proposition}

\par Given a smooth family $\Psi:[0,1]\times D\to F$ of deformations of the integrable distribution $D$, i.e. $\Psi$ is smooth and for any $t\in[0,1]$, the graph of the bundle map $\Psi_t=\Psi|_{\{t\}\times D}:D\to F$ is an integrable distribution. An infinitesimal deformation of the foliation $\mathcal{F}$ is defined to be a bundle map of the form $\frac{d \Psi_t}{dt}|_{t=0}$. It is well-known the first order cohomology of a flat $L_\infty[1]$-algebra describes the first order obstruction of the deformation. For more details on the obstruction theory of deformations see \cite{kieserman}.

\begin{proposition}
The graph of a bundle map $\psi:D\to F$ gives an infinitesimal deformation of the integrable distribution $D$ if and only if $\psi$ is closed in the cochain complex $(\mathbb{L},m_1)$ (see \ref{cohomology}), in other words, $m_1(\psi)=0$ or 
\[
	\pi_F([\psi(v),w]+[v,\psi(w)])=\psi([v,w])
\]
for any $v,w\in \Gamma(D)$.
\end{proposition} 

\par Huebschmann (\cite{huebschmann}) first proved that a Lie subalgebroid in a Lie algebroid gives rise to a "quasi-Lie-Rinehart algebra", which is a homotopy version of a Lie-Rinehart algebra. More recently, Vitagliano (\cite{luca}) remarked that quasi-Lie-Rinehart algebras are strong homotopy Lie-Rinehart algebras and described the strong homotopy Lie-Rinehart algebra of a foliation in terms of Frolicher-Nijenhuis calculus. His approach applies to Lie subalgebroids in general as well. Finally, he presented the brackets in a strong homotopy Lie-Rinehart algebra as higher derived brackets, and, in this way, he could relate Huebschmann's results to ours.

\subsection{Deformation of Complex Submanifolds}
\par Suppose $X$ is a complex manifold with the canonical almost complex structure $J\in \Gamma(T^*X \otimes TX)$ associated with its complex structure. Let $T_{\mathbb{C}}X=TX\otimes\mathbb{C}$ and $\mathcal{J}$ the $\mathbb{C}$-linear extension of $J$. $T_{\mathbb{C}}X=T^{1,0}X\oplus T^{0,1}X$ where $T^{1,0}X$ and $T^{0,1}X$ are the $\mathcal{J}$-eigenbundles of $i$ and $-i$ separately. Let $\pi_{1,0}:T_{\mathbb{C}}X\to T^{1,0}X$ and $\pi_{0,1}:T_{\mathbb{C}}X\to T^{0,1}X$ stand for the projections.

\par If deforming the almost complex structure $J$ into some $J'\in \Gamma(T^*X \otimes TX)$, the $\mathbb{C}$-linear extension $\mathcal{J}'$ of $J'$ determines another decomposition $T_{\mathbb{C}}X=(T^{1,0}X)'\oplus (T^{0,1}X)'$ which satisfies $(T^{1,0}X)'$ and $(T^{0,1}X)'$ are conjugate to each other, i.e. $\overline{(T^{1,0}X)'}=(T^{0,1}X)'$. The almost complex structure $J'$ and the decomposition $T_{\mathbb{C}}X=(T^{1,0}X)'\oplus (T^{0,1}X)'$ are in one-one correspondence. As $J'$ is not far away from $J$, the map $(T^{0,1}X)'\overset{\pi_{0,1}}{\longrightarrow} T^{0,1}X$ is an isomorphism and $J$ determines a map $\xi_{J'}:T^{0,1}X\overset{(\pi_{0,1})^{-1}}{\longrightarrow} (T^{0,1}X)'\overset{\pi_{1,0}}{\longrightarrow}T^{1,0}X$. Here $\xi_{J'}$ can be treated as an element in $\Gamma((T^{0,1}X)^* \otimes T^{1,0}X)$. Conversely, given an element $\xi\in\Gamma((T^{0,1}X)^* \otimes T^{1,0}X)$, $gr(\xi)$ and $\overline{gr(\xi)}$ forms a decomposition of $T_{\mathbb{C}}X$ if and only if
\begin{eqnarray}
\label{invertible}
	Id-\bar{\xi}\circ \xi:T^{0,1}X\to T^{0,1}X\text{ is invertible.}
\end{eqnarray}
Here $\bar{\xi}$ is the conjugate of $\xi$. Besides, a complex structure is just an almost complex structure which is integrable, so we make the following definition.

\begin{definition}
A deformation of the almost complex structure $J$ is an element $\xi\in\Gamma((T^{0,1}X)^* \otimes T^{1,0}X)$ satisfying \ref{invertible}. A deformation of the complex structure of $X$ is a deformation of the almost complex structure $J$ which is integrable.
\end{definition}

\par $T_{\mathbb{C}}X$ is a complex Lie algebroid(\cite{ping2}). By the Newlander-Nirenberg theorem(\cite{kodaira1}), the almost complex structure $J'$ is integrable if and only if $gr(\xi_{J'})=(T^{0,1}X)'$ is a Lie subalgebroid of $T_{\mathbb{C}}X$, so $\xi\in\Gamma((T^{0,1}X)^* \otimes T^{1,0}X)$ satisfying (\ref{invertible}) is a deformation of the complex structure on $X$ if and only if $gr(\xi)$ is a Lie subalgebroid. Notice $T^{0,1}X$ is a Lie subalgebroid as $J$ is integrable, so deforming the complex structure $J$ can be translated into deforming the Lie subalgebroids $T^{0,1}X\subset T_{\mathbb{C}}X$. Since $T^{0,1}X$ is over the whole base space $X$, we can apply Corollary \ref{cor4} (see Remark \ref{complex}) to get  $\Omega^{0,1}(X,T^{1,0}X)=\Gamma(\sum_{n=0}^{\infty}\wedge^n(T^{0,1}X)^*\otimes T^{1,0}X)$ is a flat $L_\infty$-algebra which controls the deformation of the complex structure. Moreover except $m_k=0$ for $k\ge 4$, $m_3$ also vanishes in this special case.
\begin{proposition}
There is a differential graded Lie algebra structure on $\Omega^{0,1}(X,T^{1,0}X)$ with differential $d=\bar{\partial}$:
\[
	\tau^j_{i_1,\cdots, i_k}d\bar{z}^{i_1}\wedge\cdots\wedge d\bar{z}^{i_k}\otimes\frac{\partial}{\partial z^j}\mapsto \frac{\partial\tau^j_{i_1,\cdots, i_k}}{\partial \bar{z}^l} d\bar{z}^l\wedge d\bar{z}^{i_1}\wedge\cdots\wedge d\bar{z}^{i_k}\otimes\frac{\partial}{\partial z^j}
\]
and the bracket $[\cdot,\cdot]:\Gamma(\wedge^k(T^{0,1}X)^*\otimes T^{1,0}X)\times \Gamma(\wedge^l(T^{0,1}X)^*\otimes T^{1,0}X)\to \Gamma(\wedge^{k+l}(T^{0,1}X)^*\otimes T^{1,0}X)$ defined by
\[
	[\phi,\psi]=d\bar{z}^I\wedge d\bar{z}^J\otimes[\phi_I^k(z,\bar{z})\frac{\partial}{\partial z^k},\psi_J^l(z,\bar{z})\frac{\partial}{\partial z^l}]
\]
Moreover $\xi\in\Gamma((T^{0,1}X)^* \otimes T^{1,0}X)$ satisfying (\ref{invertible}) is a deformation of the complex structure if and only if
\[
	\bar{\partial}\xi+\frac{1}{2}[\xi,\xi]=0.
\]
\end{proposition} 

\par This is a classical result of K.Kadaira. For a systematic development of the deformation theory of complex manifolds, see \cite{kodaira1}, \cite{kodaira}.

\subsection{Deformation of homomorphisms of Lie algebroids}
\par Let $(A,[\cdot,\cdot]_1,\rho_1)$ and $(B,[\cdot,\cdot]_2,\rho_2)$ be two Lie algebroids over $M$ and $N$ separately. A homomorphism of Lie algebroids from $A$ to $B$ is a bundle map $\Psi:A\to B$ with a smooth map $\psi:M\to N$ induced by $\Psi$ satisfying
\begin{itemize}
\item $\Psi$ induces a homomorphism of Lie algebras $\Gamma(A)\to \Gamma(B)$ (also denoted by $\Psi$);
\item The following diagram commutes
\[
\begin{CD}
		\Gamma(A) @>\Psi>> \Gamma(B)\\
		@V\rho_1VV	   @VV\rho_2V\\
		TM  @>\psi_*>> TN.
\end{CD}
\] 
\end{itemize}

\par $A\oplus B$ is a Lie algebroid over $M\times N$ with bracket $[\cdot,\cdot]=[\cdot,\cdot]_1+[\cdot,\cdot]_2$ and anchor $\rho=\rho_1\times \rho_2$, i.e. for any $a,a'\in\Gamma(A)$ and $b,b'\in\Gamma(B)$
\begin{eqnarray*}
	[(a,b),(a',b')]=([a,a']_1,[b,b']_2)\\
	\rho(a,b)=(\rho(a),\rho(b))\in T(M\times N)
\end{eqnarray*}
If the homological vector field of $A[1]$ and $B[1]$ are $X_Q^1$ and $X_Q^2$ separately, the homological vector field of $(A+B)[1]$ is $X_Q=X_Q^1+X_Q^2$. A deformation of the Lie algebroid structure on $A\oplus B$ is equivalent to a deformation of the homological vector field $X_Q$ which must be of the form $\tilde{X}_Q^1+\tilde{X}_Q^2$ where $\tilde{X}_Q^1\in\mathfrak{X}(A[1])$, $\tilde{X}_Q^2\in\mathfrak{X}(B[1])$ both of which are of degree $1$. By Proposition \ref{dgla} and Theorem \ref{thm1}
\begin{proposition}
$(\mathfrak{X}((A\oplus B)[1]),[[\cdot,\cdot]],[[X_Q,\cdot]])$ is a differential grade Lie algebra which governs the deformation of the Lie algebroid structure $X_Q$ on $A\oplus B$, i.e. $\tilde{X}_Q^1+\tilde{X}_Q^2$ deforms $X_Q$ if and only if
\[
	[[X_Q,\tilde{X}_Q^1+\tilde{X}_Q^2]]+\frac{1}{2}[[\tilde{X}_Q^1+\tilde{X}_Q^2,\tilde{X}_Q^1+\tilde{X}_Q^2]]=0
\]
which is equivalent to $\tilde{X}_Q^1$ is a deformation of $X_Q^1$ and $\tilde{X}_Q^2$ is a deformation of $X_Q^2$.
\end{proposition}
Here $[[\cdot,\cdot]]$ is the bracket of super vector fields as in Section \ref{sec3}.

\par It is immediate that $\Psi:A\to B$ is a homomorphism of Lie algebroids is equivalent to $gr(\Psi)\subset A\oplus B$ is a Lie subalgebroid over the submanifold $gr(\psi)\subset M\times N$. Denote the subbundle $gr(\Psi)$ by $E$ and the submanifold $gr(\psi)$ by $S$. Then the pull-back of $B$ along the natural projection $M\times N\to N$ restricting on $S$ is a complement of $E$ in $(A\oplus B)|_S$. Denote it by $F$ and we have $(A\oplus B)|_S= E\oplus F$. Identify $NS$ as a tubular neighborhood of $S$ in $M\times N$ and denote the pull-back of $E$ and $F$ along $NS\to S$ by $\tilde{E}$ and $\tilde{F}$ as before and we can identify $M\times N$ with $NS$, $A\oplus B$ with $\tilde{E}\oplus \tilde{F}$. Let $\tilde{\Psi}:A\to B$ be another homomorphism of Lie algebroids over $\tilde{\psi}:M\to N$. We say $\tilde{\Psi}$ is a deformation of $\Psi$ if
\begin{itemize}
\item $gr(\tilde{\psi})=gr(\sigma)$ for some $\sigma\in\Gamma(NS)$
\item $gr(\tilde{\Psi})=gr(\phi)$ for some $\phi: \tilde{E}|_{S_\sigma}\to \tilde{F}|_{S_\sigma}$
\item $gr(\tilde{\Psi})$ is a Lie subalgebroid of $A\oplus B$ over $gr(\tilde{\psi})$.
\end{itemize}
Thus a bundle map $\tilde{\Psi}:A\to B$ not far away from $\Psi:A\to B$ is equivalent to a pair $(\sigma,\phi)\in \Gamma(NS\oplus (\tilde{E}^*|_{S_\sigma}\otimes \tilde{F}|_{S_\sigma}))$. Let $p^*(NS\oplus F)$ denote the pull-back of $NS\oplus F$ over $S$ along $E\to S$ and $\mathfrak{a}=\Gamma(p^*NS\oplus p^*F[1])$. As in Section \ref{sec4}, the pair $(\sigma,\phi)$ corresponds to a unique $X_{\sigma,\phi}\in \mathfrak{a}$. Denote the map $\tilde{\Psi}\mapsto X_{\sigma,\psi}$ by $h$. Notice $X_Q^1$ is already analytic along the fibers of $NS$, by Theorem \ref{thm2} we have
\begin{proposition}
\label{homo1}
$(\mathfrak{X}(A\oplus B)[1],[[\cdot,\cdot]],\mathfrak{a}, P)$ is a V-algebra and $\mathfrak{a}_{X_Q}^P$ is a flat $L_\infty[1]$-algebra. If $X_Q^2$ is analytic along the fibers of $NS$, the $L_\infty[1]$-algebra $\mathfrak{a}_{X_Q}^P$ governs the deformation of the Lie algebroid homomorphism $\Psi:A\to B$. A bundle map $\tilde{\Psi}:A\to B$ over $\tilde{\psi}:M\to N$ is a deformation of $\Psi$ if and only if the vector field $X_{\sigma,\phi}=h(\tilde(\Psi))$ is a Maurer-Cartan element of $\frak{a}_{X_Q}^P$.
\end{proposition}

\par If we only consider the deformation $\tilde{\Psi}:A\to B$ over the same base map $\psi:M\to N$, i.e. $\tilde{\psi}=\psi$, $X_Q^2$ is not required being analytic and the $L_\infty[1]$-algebra can be simplified the same as in Corollary \ref{cor1}.

\par We can also govern the simultaneous deformation by Corollary \ref{cor2}.
\begin{proposition}
\label{homo2}
If $X_Q^2$ is analytic along the fibers of $NS$ and $\tilde{\Psi}$ is a deformation of $\Psi$ with $X_{\sigma,\phi}=h(\tilde(\Psi))\in\mathfrak{a}$, the flat $L_\infty[1]$-algebra $((\mathfrak{X}(A\oplus B)[1])[1]\oplus \mathfrak{a})_{X_Q}^{P_{I(X_{\sigma,\psi})}}$ governs the simultaneous deformation of the Lie algebroid structures on $A$, $B$ and Lie algebroid homomorphisms between them, i.e. $\tilde{X}_Q^1$, $\tilde{X}_Q^2$ deforms the Lie algebroid structures of $A$, $B$ and $\tilde{\Psi}':A\to B$ such that $\tilde{\Psi}+\tilde{\Psi}'$ is a homomorphism of Lie algebroids with deformed Lie algebroids structures if and only if $(\tilde{X}_Q^1+\tilde{X}_Q^2,h(\tilde{\Psi}'))$ is a Maurer-Cartan element of $((\mathfrak{X}(A\oplus B)[1])[1]\oplus \mathfrak{a})_{X_Q}^{P_{I(X_{\sigma,\psi})}}$.
\end{proposition}

\par Specially we can choose $\tilde{\Psi}=0$ (trivial deformation of $\Psi$), then $\tilde{X}_Q^1$, $\tilde{X}_Q^2$ deforms the Lie algebroid structures of $A$, $B$ and $\tilde{\Psi}':A\to B$ is a homomorphism of Lie algebroids with deformed Lie algebroids structures if and only if $(\tilde{X}_Q^1+\tilde{X}_Q^2,h(\tilde{\Psi}'))$ is a Maurer-Cartan element of $((\mathfrak{X}(A\oplus B)[1])[1]\oplus \mathfrak{a})_{X_Q}^{P}$.

\begin{remark}
If the Lie algebroids $A$ and $B$ degenerate to Lie algebras, then $\Psi$ degenerates to homomorphism of Lie algebras. Our Proposition \ref{homo1} and \ref{homo2} recover the results for Lie algebra homomorphisms in \cite{zambon}.
\end{remark}

\end{document}